\documentclass{eptcs}
\providecommand{\event}{QPL 2011} % Name of the event you are submitting to

\sffamily

\usepackage{amssymb,amsfonts,latexsym,stmaryrd}
\usepackage{amsmath}
\usepackage[PostScript=dvips]{diagrams}
\usepackage{euscript}
\usepackage{QED}

\newtheorem{theorem}{Theorem}[section]

\newtheorem{proposition}[theorem]{Proposition}
\newtheorem{conjecture}[theorem]{Conjecture}

\newcommand{\myemph}{\textbf}

\newcommand{\Cech}{\v{C}ech~}

\newcommand{\ZZ}{\mathbb{Z}}
\newcommand{\FF}{\mathcal{F}}
\newcommand{\UU}{\EuScript{U}}
\newcommand{\rmap}[2]{\rho^{#1}_{#2}}
\newcommand{\lrarr}{\longrightarrow}
\newcommand{\rarr}{\rightarrow}
\newcommand{\id}{\mathsf{id}}
\newcommand{\NN}{\mathsf{N}}
\newcommand{\vn}{\varnothing}
\newcommand{\bdmap}{\partial}
\newcommand{\cobd}{\delta}
\newcommand{\Ccoch}[1]{C^{#1}(\UU, \FF)}
\newcommand{\Cocyc}[1]{Z^{#1}(\UU, \FF)}
\newcommand{\Cobound}[1]{B^{#1}(\UU, \FF)}
\newcommand{\Cohom}[1]{\check{H}^{#1}(\UU, \FF)}
\newcommand{\EE}{\mathcal{E}}

\newcommand{\FR}{F_{R}}
\newcommand{\FZ}{F_{\ZZ}}
\newcommand{\FC}{\FF_{\bar{C}_{1}}}
\newcommand{\FU}{\FF_{\bar{U}}}
\newcommand{\FUi}{\FF_{\bar{U}_{i}}}
\newcommand{\CoFC}[1]{\check{H}^{#1}(\UU, \FC)}

\newcommand{\CoFUi}[1]{\check{H}^{#1}(\UU, \FUi)}
\newcommand{\Set}{\mathbf{Set}}
\newcommand{\ie}{\textit{i.e.}~}
\newcommand{\supp}{\mathsf{supp}}

\newcommand{\Real}{\mathbb{R}}

\newcommand{\obst}{\gamma}

\newcommand{\Cij}{C_{i,j}}

\newcommand{\GCD}{\textsf{GCD}~}

\title{The Cohomology of Non-Locality and Contextuality}
\author{Samson Abramsky
\qquad
Shane Mansfield
\qquad Rui Soares Barbosa
\institute{Department of Computer Science\\
University of Oxford
}
\email{\{samson.abramsky,shane.mansfield,rui.soaresbarbosa\}@cs.ox.ac.uk}
}
\def\titlerunning{The Cohomology of Non-Locality and Contextuality}
\def\authorrunning{Samson Abramsky, Shane Mansfield \& Rui S. Barbosa}

\begin{document}
\maketitle

\begin{abstract}
In a previous paper with Adam Brandenburger, we used sheaf theory to analyze the structure of non-locality and contextuality. 
Moreover, on the basis of this formulation, we showed that the phenomena of non-locality and contextuality can be characterized precisely in terms of obstructions to the existence of global sections. 

Our aim in the present work is to build on these results, and to use the powerful tools of sheaf cohomology to study the structure of non-locality and contextuality. 
We use the Cech cohomology on an abelian presheaf derived from the support of a probabilistic model, viewed as a compatible family of distributions, in order to define a cohomological obstruction for the family as a certain cohomology class. 
This class vanishes if the family has a global section. 
Thus the non-vanishing of the obstruction provides a sufficient (but not necessary) condition for the model to be contextual. 

We show that for a number of salient examples, including PR boxes, GHZ states, the Peres-Mermin magic square, and the 18-vector configuration due to Cabello et al.~giving a proof of the Kochen-Specker theorem in four dimensions, the obstruction does not vanish, thus yielding cohomological witnesses for contextuality.
\end{abstract}

\section{Introduction}

Non-locality and contextuality are fundamental features of physical theories, which contradict the intuitions underlying classical physics. They are, in particular, prominent features of quantum mechanics, and the goal of the  classic no-go theorems by Bell \cite{bell1964einstein}, Kochen-Specker \cite{kochen1975problem}, et al.~is to show that they are \textit{necessary features} of any theory whose experimental predictions agree with those of quantum mechanics.

Bell's insights into non-locality have been seminal to the current developments in quantum information, where entanglement is viewed as a key informatic resource; and there has also been considerable recent work on experimental tests for contextuality \cite{bartosik2009experimental,kirchmair2009state}.

In a previous paper with Adam Brandenburger \cite{abramsky2011unified}, we used the mathematics of \myemph{sheaf theory} to analyze the structure of non-locality and contextuality. Sheaf theory is pervasive in modern mathematics, allowing the passage from local to global \cite{mac1992sheaves}. Starting from a simple experimental scenario, and the kind of probabilistic models familiar from discussions of Bell's theorem, Popescu-Rohrlich boxes \cite{popescu1994quantum}, etc., we gave a very direct, compelling formalization of these notions in sheaf-theoretic terms.
Moreover, on the basis of this formulation, we showed that the phenomena of non-locality and contextuality can be characterized precisely  in terms of \myemph{obstructions to the existence of global sections}.

Our aim in the present work is to build on these results, and to use the powerful tools of \myemph{sheaf cohomology} to study the structure of non-locality and contextuality. The present paper describes work in progress, and only represents an initial step in this direction. Nevertheless, enough has been achieved to indicate that this approach has some promise, and merits further investigation.

We briefly summarize our results:
\begin{itemize}
\item We use the formalization of no-signalling probabilistic models as compatible families of sections on a presheaf of distributions developed in  \cite{abramsky2011unified}; compatibility corresponds precisely to the no-signalling condition. The family is defined on a cover corresponding to the sets of compatible measurements.
\item The locality/non-contextuality  of the model corresponds to the existence of a global section for this family, as shown in \cite{abramsky2011unified}.
\item We use the \Cech cohomology on an abelian presheaf derived from the support of the model in order to define a \myemph{cohomological obstruction} for the family as a certain cohomology class.
This class vanishes if the family has a global section.
Thus the non-vanishing of the obstruction provides a sufficient (but not necessary) condition for the model to be contextual.
\item We show that for a number of salient examples, including PR boxes  \cite{popescu1994quantum}, GHZ states \cite{greenberger1989going,greenberger1990bell},  the Peres-Mermin magic square \cite{peres1990incompatible,mermin1990simple}, and the 18-vector configuration giving a proof of the Kochen-Specker theorem in four dimensions from \cite{cabello1996bell}, the obstruction does not vanish, thus yielding cohomological witnesses for contextuality.
\end{itemize}

The further contents of the paper are as follows. We review the sheaf formulation from  \cite{abramsky2011unified} in Section~2, and \Cech cohomology in Section~3. We define the cohomological obstruction in Section~4, and consider various examples in Sections~5 and ~6. Finally, limitations of the current results and further directions are discussed in Section~7.

\section{Sheaf Formulation of Measurement Scenarios}

We recall the basic ideas of \cite{abramsky2011unified}.

We work over a finite discrete space $X$, which we think of as a set of \myemph{measurement labels}.
We fix a  cover $\UU = \{ C_1 , \ldots , C_n \}$, with $\bigcup \UU = X$,  which represents the set of \myemph{compatible families of measurements}, \ie those which can be made jointly. Fixing a finite set $O$ of \myemph{outcomes}, we have the presheaf of sets $\EE$ on $X$, where
$\EE(U) := O^U$, and restriction is simply function restriction:
given $U \subseteq U'$, 
\[ \rmap{U'}{U} : \EE(U') \rarr \EE(U) :: s \mapsto s | U . \]
Since $X$ is discrete, $\EE$ is (trivially) a sheaf. We think of it as the sheaf of \myemph{events}.

An empirical model $e$ in the sense of  \cite{abramsky2011unified} is a compatible family $\{ e_C \}_{C \in \UU}$, where $e_C$ is a probability distribution on $\EE(C)$. Here compatibility uses the definition of restriction on distributions, which we omit since we shall not need it. The support of $e$ determines a sub-presheaf $S_e$ of $\EE$: 
\[ S_e(U) := \{ s \in \EE(U) \mid s \in \supp(e_U) \} . \]
Here $e_U = e_C | U$ for any $C \in \UU$ such that $U \subseteq C$. The compatibility of the family $\{ e_C \}$ ensures that this is independent of the choice of $C$.

We have the following notions from \cite{abramsky2011unified}.
\begin{itemize}
\item The model $e$ is \myemph{possibilistically extendable} iff for every $s \in S_e(C)$, $s$ is a member of a compatible family $\{ s_i \in S_e(C_i) \}_{i = 1}^{n}$.
It is \myemph{possibilistically non-extendable} if for some $s$, there is no such family.
\item The model $e$ is \myemph{strongly contextual} if for every $s$ there is no such family.
\end{itemize}

The results from \cite{abramsky2011unified} show that if a model is local or non-contextual in the usual sense, then it is possibilistically extendable. Thus possibilistic non-extendability is a sufficient condition for \myemph{non-locality} or \myemph{contextuality}. Strong contextuality is a much stronger condition. Thus these properties witness strong forms of the non-classical behaviour exhibited by quantum mechanics.

\section{\Cech Cohomology of a Presheaf}

We are given the following:
\begin{itemize}
\item A topological space $X$.
\item An open cover $\UU$ of $X$.
\item A presheaf $\FF$ of abelian groups on $X$.
For each open set $U$ of $X$, $\FF(U)$ is an abelian group, and when $U \subseteq V$, there is a group homomorphism $\rmap{V}{U} : \FF(V) \rarr \FF(U)$. These assignments are functorial:
$\rmap{U}{U} = \id_{U}$, and if $U \subseteq U' \subseteq U''$, then
\[ \rmap{U'}{U} \circ \rmap{U''}{U'} = \rmap{U''}{U} . \]
\end{itemize}
The \myemph{nerve} $\NN(\UU)$ of the cover $\UU$ is defined to be the abstract simplicial complex comprising those finite subsets of $\UU$ with non-empty intersection. Concretely, we take a $q$-simplex to be a list $\sigma = (U_0, \ldots , U_q)$ of elements of $\UU$, with $| \sigma | := \cap_{i=0}^q U_i \neq \vn$. Thus a $0$-simplex $(U)$ is a single element of the cover $\UU$.
We write $\NN(\UU)^q$ for the set of $q$-simplices.

%We shall also consider the \myemph{augmentation} of this complex with a (unique)  $-1$-simplex $()$; we define $| () | := X$, which agrees with the standard definition of the intersection of an empty family.

Given a $q+1$-simplex $\sigma = (U_0, \ldots , U_{q+1})$, there are  $q$-simplices
\[ \bdmap_j (\sigma) :=  (U_0, \ldots, \widehat{U_j}, \ldots , U_{q+1}), \qquad 0 \leq j \leq q \]
obtained by omitting one of the elements of the $q+1$-simplex. Note that:
\[ | \sigma | \; \subseteq \; | \bdmap_j(\sigma) | . \]

We shall now define the \myemph{\Cech cochain complex}.
For each $q \geq 0$, we define the abelian group $\Ccoch{q}$:
\[ \Ccoch{q} \; := \; \prod_{\sigma \in \NN(\UU)^q} \FF( | \sigma |) . \]
%Note that $\Ccoch{-1} = \FF(X)$.

We also define the \myemph{coboundary maps} 
\[ \cobd^{q} : \Ccoch{q} \lrarr \Ccoch{q+1} . \]
For $\omega = (\omega(\tau))_{\tau \in \NN(\UU)^q} \in \Ccoch{q}$, and $\sigma \in \NN(\UU)^{q+1}$, we define:
\[ \cobd^{q}(\omega)(\sigma) \; := \; \sum_{j = 0}^{q} (-1)^j \rmap{|\bdmap_j(\sigma)|}{|\sigma|}\omega(\bdmap_j \sigma) . \]
%Note that, for $q = -1$:
%\[ \cobd^{-1} : \FF(X) \lrarr \prod_{U \in \UU} \FF(U) :: s \mapsto (\rmap{X}{U}(s))_{U \in \UU} . \]
For each $q$, $\cobd^q$ is a group homomorphism.

We shall also consider the \myemph{augmented complex} $\mathbf{0} \rarr \Ccoch{0} \rarr \cdots$.

\begin{proposition}
\label{boundprop}
For each $q$, $\cobd^{q+1} \circ \cobd^q = 0$.
\end{proposition}

We define $\Cocyc{q}$, the \myemph{$q$-cocycles}, to be the kernel of $\cobd^{q}$. We define $\Cobound{q}$, the \myemph{$q$-coboundaries}, to be the image of $\cobd^{q-1}$. These are subgroups of $\Ccoch{q}$, and by Proposition~\ref{boundprop}, $\Cobound{q} \subseteq \Cocyc{q}$.
We define the \myemph{$q$-th \Cech cohomology group} $\Cohom{q}$ to be the quotient group 
\[ \Cohom{q} \; := \; \Cocyc{q} / \Cobound{q} . \]
Note that $\Cobound{0} = \mathbf{0}$, so $\Cohom{0} \cong \Cocyc{0}$.

Given a cocycle $z \in \Cocyc{q}$, the \myemph{cohomology class} $[z]$ is the image of $z$ under the canonical map
\[ \Cocyc{q} \lrarr \Cohom{q} . \]

A \myemph{compatible family} with respect to a cover $\UU$ is a family $\{ r_i \in \FF(U_i) \}$ for $U_i \in \UU$, such that, for all $i$, $j$:
\[ r_i | U_i \cap U_j = r_j | U_i \cap U_j . \]
\begin{proposition}
\label{compfamprop}
There is a bijection between compatible families and elements of the zeroth cohomology group $\Cohom{0}$.
\end{proposition}
\begin{proof}
Cochains $c = (r_i)_{U_i \in \UU}$ in $\Ccoch{0}$ correspond to families $\{ r_i \in \FF(U_i) \}$.
For each $1$-simplex $\sigma = (C_i,C_j)$, 
\[ \cobd^{0}(c)(\sigma) \; = \; r_i | C_i \cap C_j \; - \; r_j | C_i \cap C_j . \]
Hence $\cobd^{0}(c) = 0$ if and only if the corresponding family is compatible.
\end{proof}

We shall also use the  \textit{relative cohomology} of $\FF$ with respect to an open subset $U \subseteq X$.

We define two auxiliary presheaves related to $\FF$. Firstly, $\FF | U$ is defined by
\[ \FF | U (V) := \FF(U \cap V) . \]
There is an evident presheaf morphism 
\[ p : \FF \lrarr \FF | U :: p_V : r \mapsto r | U \cap V . \]
Then $\FU$ is defined by $\FU(V) := \ker(p_V)$. Thus we have an exact sequence of presheaves
\[ \mathbf{0} \rTo \FU \rTo \FF \rTo^{p} \FF | U . \]
The relative cohomology of $\FF$ with respect to $U$ is defined to be the cohomology of the presheaf $\FU$.

We have the following refined version of Proposition~\ref{compfamprop}.
\begin{proposition}
\label{relcompfamprop}
For any $U_i \in \UU$,
the elements of the  relative cohomology group $\CoFUi{0}$ correspond bijectively to compatible families $\{ r_j \}$ such that $r_i = 0$.
\end{proposition}
\begin{proof}
By Proposition~\ref{compfamprop}, compatible families correspond to cocycles 
$r = (r_j)$ in $\Ccoch{0}$.  By compatibility, $r_i | C_i \cap C_j =  r_j | C_i \cap C_j$ for all $j$.
Hence $r$ is  in $C^{0}(\UU, \FUi)$ if and only if $r_i = p_{U_{i}}(r_i) = 0$.
\end{proof}

\section{Cohomological Obstructions}

Given a commutative ring $R$, we define a functor $\FR : \Set \lrarr \Set$.
For any set $X$, the \myemph{support} $\supp(\phi)$ of a function $\phi : X \rarr R$ is the set of $x \in X$ such that $\phi(x) \neq 0$. We define $\FR(X)$ to be the set of functions $\phi : X \rarr R$ of finite support. There is an embedding $x \mapsto 1 \cdot x$ of $X$ in $\FR(X)$, which we shall use implicitly throughout.

Given $f : X \rarr Y$, we define:
\[ \FR f : \FR X \lrarr \FR Y :: \phi \mapsto [y \mapsto  \sum_{f(x) = y} \phi(x)] . \]
This assignment is easily seen to be functorial.

In fact, $\FR(X)$ is the \myemph{free $R$-module generated by $X$}, and in particular, it is an abelian group; while $\FR(f)$ is a group homomorphism for any function $f$. In particular, taking $R = \ZZ$, $\FZ(X)$ is the \myemph{free abelian group generated by $X$}.

Thus, given any presheaf of sets $P$ on $X$, we obtain a presheaf of abelian groups $\FZ P$ by composition: $\FZ P(U) := \FZ(P(U))$.

Given an empirical model $e$ defined on the cover $\UU$, we shall work with the \Cech cohomology groups $\Cohom{q}$ for the abelian presheaf $\FF := \FZ S_e$. Note that, for any set of measurements $U$, $\FF(U)$ is the set of \myemph{formal $\ZZ$-linear combinations of sections} in the support of $e_U$.

To each $s \in S_e(C)$, we shall associate an element $\obst(s)$ of a cohomology group, which can be regarded as an obstruction to $s$ having an extension within the support of $e$ to a global section. In particular, the existence of such an extension implies that the obstruction vanishes. In good cases, these two conditions are equivalent, yielding  \myemph{cohomological characterizations} of contextuality and strong contextuality.

For notational convenience, we shall fix an element $s = s_1 \in S_e(C_1)$. Because of the compatibility of the family $\{ e_C \}$, which is equivalent to no-signalling \cite{abramsky2011unified}, there is a family $\{ s_i \in S_e(C_i) \}$ with $s_1 | C_1 \cap C_i = s_i | C_1 \cap C_i$, $i = 2, \ldots , n$.

%The following construction is adapted from \cite{beniaminov1995algebraic}.

We define the cochain $c := (s_1, \ldots , s_n) \in \Ccoch{0}$. The coboundary of this cochain is $z := \cobd^0(c)$.

\begin{proposition}
The coboundary $z$ of $c$  vanishes under restriction to $C_1$, and hence is a cocycle in the relative cohomology with respect to $C_1$.
\end{proposition}
\begin{proof}
We write $\Cij := C_i \cap C_j$. For all $i, j$, we define $z_{i,j} := z(\Cij) = s_i | \Cij  - s_j | \Cij$. Because of the no-signalling assumption on the family $\{ s_i \}$, for all $i, j$,
\[ s_i | C_1 \cap \Cij = (s_1 | C_1 \cap C_i) | C_j = s_1 | C_1 \cap \Cij . \]
Similarly, $s_j | C_1 \cap \Cij =  s_1 | C_1 \cap \Cij$. Hence $z_{i,j} | C_1 \cap \Cij = 0$, and 
$z_{i,j}  \in \FC(C_i \cap C_j)$.
Thus $z = (z_{i,j})_{i,j} \in C^{1}(\UU, \FC)$. 

Note that $\cobd^{1} : C^{1}(\UU, \FC) \rarr C^{2}(\UU, \FC)$ is the restriction of the coboundary map on $\Ccoch{1}$. Hence  $z = \cobd^0(c)$  is a cocycle.
\end{proof}

We  define $\obst(s_1)$ as the cohomology class $[z] \in \CoFC{1}$.

\paragraph{Remark}
There is a more conceptual way of defining this obstruction, using the connecting homomorphism from the long exact sequence of cohomology; see \cite{ghrist2011applications}.
We have given a more concrete formulation, which may be easier to grasp, and is also convenient for computation.

\vspace{.1in}
Note that, although  $z = \cobd^0(c)$, it  is not necessarily a coboundary in $C^{1}(\UU, \FC)$, since $c$ is not a cochain in  $C^{0}(\UU, \FC)$, as $p_{C_i}(s_i) = s_i | C_1 \cap C_i  \neq 0$. Thus in general, we need not have $[z] = 0$.

\begin{proposition}
\label{obstprop}
The following are equivalent:
\begin{enumerate}
\item The cohomology obstruction vanishes: $\obst(s_1) = 0$.
\item There is a family $\{ r_i \in \FF(C_i) \}$ with $s_1 = r_1$, and for all $i, j$:
\[ r_i | C_i \cap C_j = r_j | C_i \cap C_j . \]
\end{enumerate}
\end{proposition}
\begin{proof}
The obstruction vanishes if and only if there is a cochain $c' = (c'_1, \ldots , c'_n) \in C^{0}(\UU, \FC)$  with $\cobd^{0}(c') = \cobd^0(c)$, or equivalently $\cobd^{0}(c - c')  = 0$, \ie such that $c - c'$ is a cocycle.
By Proposition~\ref{compfamprop}, this is equivalent to $\{ r_i := s_i - c'_i \}$ forming a compatible family.
Moreover, $c' \in C^{0}(\UU, \FC)$ implies $c'_1 = p_{C_1}(c'_1) = 0$, so  $r_1 = s_1$.

For the converse, suppose we have a family $\{ r_i \in \FF(C_i) \}$ as in (2).
We define $c' := (c'_1, \ldots , c'_n)$, where $c'_i := s_i - r_i$. 
Since $r_1 = s_1$, $p_{C_i}(c'_i) = s_1 | C_{1,i} - r_1 | C_{1,i} = 0$ for all $i$, and $c' \in C^{0}(\UU,\FC)$.
We must show that $\cobd^{0}(c') = z$, \ie that $z_{i,j} = c'_i| \Cij - c'_j | \Cij$. This holds since $r_i | \Cij = r_j | \Cij$.
%Assume that $\obst(s_1) = 0$. Then there is a cochain $c' = (c'_1, \ldots , c'_n) \in C^{0}(\UU, \FC)$ with $\cobd^{0}(c') = z$. Thus $z_{i,j} = c'_i | \Cij - c'_j | \Cij$.
%We define $r_i := s_i - c'_i$.
%Since $c' \in C^{0}(\UU, \FC)$, $c'_1 = p_{C_1}(c'_1) = 0$, so $r_1 = s_1$.
%Also, 
%\[ \begin{array}{lcl}
%r_i | \Cij & = & s_i | \Cij - c'_i | \Cij \\
%& = & s_i | \Cij - z_{i,j}  - c'_j | \Cij \\
%& = & s_j| \Cij - c'_j | \Cij \\
%& = & r_j | \Cij .
%\end{array}
%\]
\end{proof}

As an immediate application to contextuality, we have the following.

\begin{proposition}
If the model $e$ is possibilistically extendable, then the obstruction vanishes for every section in the support of the model. If $e$ is not strongly contextual, then the obstruction vanishes for some section in the support.
\end{proposition}
\begin{proof}
If $e$ is possibilistically extendable, then for every $s \in S_e(C_i)$, there is a compatible family $\{ s_j \in S_e(C_j) \}$ with $s = s_i$. Applying the embedding of $S_e(C_j)$ into $\FF(C_j)$, by Proposition~\ref{obstprop} we conclude that $\obst(s) = 0$.
The same argument can be applied to a single section witnessing the failure of strong contextuality.
\end{proof}

Thus we have a \textit{sufficient condition} for contextuality in the non-vanishing of the obstruction.
The non-necessity of the condition arises from the possibility of `false positives': families $\{ r_i \in \FF(C_i) \}$ which do not determine a \textit{bona fide} global section in $\EE(X)$.

\section{Examples}

\subsection*{Example: Hardy}

We begin with an example to show that false positives do indeed arise.

We consider the Hardy model \cite{hardy1993nonlocality}; the support is given as follows.
\begin{center}
\renewcommand{\arraystretch}{1.2}
\begin{tabular}{l|cccc} 
 &  $(0, 0)$ & $(0, 1)$ & $(1, 0)$ & $(1, 1)$ \\ \hline
$(a, b)$ &  $1$ &  $1$ &  $1$ &  $1$ \\
$(a, b')$ &   $0$ &  $1$ &  $1$ &  $1$ \\
$(a', b)$ &  $0$ &  $1$ &  $1$ & $1$ \\
$(a', b')$ &   $1$ &  $1$ &  $1$ & $0$ \\
\end{tabular}
\end{center}

\vspace{.2in}
\noindent We enumerate the sections as follows:
\begin{center}
\renewcommand{\arraystretch}{1.2}
\begin{tabular}{l||c|c|c|c|}
& $(0, 0)$ & $(0, 1)$ & $(1, 0)$ & $(1, 1)$    \\  \hline \hline
$(a, b)$ & $s_1$ & $s_2$ & $s_3$ & $s_4$  \\  \hline
$(a, b')$ & $s_5$ & $s_6$ & $s_7$ & $s_8$  \\  \hline
$(a', b)$ & $s_9$ & $s_{10}$ & $s_{11}$ & $s_{12}$  \\  \hline
$(a', b')$ & $s_{13}$ & $s_{14}$ & $s_{15}$ & $s_{16}$  \\  \hline
\end{tabular}
\end{center}

The section $s_1$ provides a witness for the non-locality of the Hardy model. It is not a member of any  compatible family of sections in the support. However, we have the following family of $\ZZ$-linear combinations of sections:
\[ r_1 = s_1, \quad r_2 = s_6 + s_7 - s_8, \quad r_3 = s_{11}, \quad r_4 = s_{15} . \]
One can check that
{\renewcommand{\arraystretch}{1.2}\[
 \begin{array}{lclcl}
r_2 | a & = &  1 \cdot (a \mapsto 0) + 1 \cdot (a \mapsto 1) - 1 \cdot (a \mapsto 1) & = &  r_1 | a , \\
r_2 | b' & = &  1 \cdot (b' \mapsto 1) + 1 \cdot (b' \mapsto 0) - 1 \cdot (b' \mapsto 1) & = & r_4 | b' . 
\end{array}
\]}
Thus this family meets the conditions of Proposition~\ref{obstprop}, and the obstruction $\obst(s_1)$ vanishes.

\subsection*{Example: PR-Box}

There is better news when we look at the PR-box:
\begin{center}
\renewcommand{\arraystretch}{1.2}
\begin{tabular}{l|ccccc}
& $(0, 0)$ & $(0, 1)$ & $(1, 0)$ & $(1, 1)$  &  \\ \hline
$(a, b)$ & $1$ & $0$ & $0$ & $1$ & \\
$(a, b')$ & $1$ & $0$ & $0$ & $1$ & \\
$(a', b)$ & $1$ & $0$ & $0$ & $1$ & \\
$(a', b')$ & $0$ & $1$ & $1$ & $0$ & 
\end{tabular}
\end{center}
This is a strongly contextual model \cite{abramsky2011unified}, so no section in the support is a member of a compatible family. 
The coefficients for a candidate family $\{ r_i \}$  can be displayed as follows:
\begin{center}
\renewcommand{\arraystretch}{1.2}
\begin{tabular}{l|ccccc}
& $00$ & $01$ & $10$ & $11$  &  \\ \hline
$AB$ & $a$ & $0$ & $0$ & $b$ & \\
$AB'$ & $c$ & $0$ & $0$ & $d$ & \\
$A'B$ & $e$ & $0$ & $0$ & $f$ & \\
$A'B'$ & $0$ & $g$ & $h$ & $0$ & 
\end{tabular}
\end{center}
The constraints arising from the requirements that  $r_i | \Cij = r_j | \Cij$ are:
\[ a=c, \quad b=d, \quad a = e, \quad b = f, \quad c = h, \quad  d = g, \quad e = g, \quad f = h . \]
These imply that all the variables are equal.

Checking that a section in the support is a member of such a family amounts to assigning $1$ to the variable labelling that section, and $0$ to the other variable in its row.
Clearly such an assignment is incompatible with the above constraints, since it implies $1 = 0$.

Hence there can be no such family, and the obstruction does not vanish for any section in the support, witnessing the strong contextuality of the PR box.

\subsection*{Example: GHZ}

The previous example suggests looking at GHZ, which is also strongly contextual, and of course is realizable in quantum mechanics.

The support for (the relevant part of) GHZ is as follows:
\begin{center}
\renewcommand{\arraystretch}{1.2}
\begin{tabular}{c|cccccccc}
&  $000$ & $001$ & $010$ & $011$  & $100$ & $101$ & $110$ & $111$  \\ \hline
$ABC$ & $1$ & $0$ & $0$ & $1$ & $0$ & $1$ & $1$ & $0$  \\
$AB'C'$ & $0$ & $1$ & $1$ & $0$ & $1$ & $0$ & $0$ & $1$  \\
$A'BC'$ & $0$ & $1$ & $1$ & $0$ & $1$ & $0$ & $0$ & $1$  \\
$A'B'C$ & $0$ & $1$ & $1$ & $0$ & $1$ & $0$ & $0$ & $1$ 
\end{tabular}
\end{center}

We display the coefficients for a candidate family as follows:
\begin{center}
\renewcommand{\arraystretch}{1.2}
\begin{tabular}{c|cccccccc}
& $000$ & $001$ & $010$ & $011$  & $100$ & $101$ & $110$ & $111$   \\ \hline
$ABC$ & $a$ & $0$ & $0$ & $b$ & $0$ & $c$ & $d$ & $0$  \\
$AB'C'$ & $0$ & $e$ & $f$ & $0$ & $g$ & $0$ & $0$ & $h$  \\
$A'BC'$ & $0$ & $i$ & $j$ & $0$ & $k$ & $0$ & $0$ & $l$  \\
$A'B'C$ & $0$ & $m$ & $n$ & $0$ & $o$ & $0$ & $0$ & $p$ 
\end{tabular}
\end{center}
The constraints arising from the requirements that  $r_i | \Cij = r_j | \Cij$ are:
{\renewcommand{\arraystretch}{1.2}
\begin{equation*}
\label{const2eq}
\begin{array}{lclclcl} 
a + b & = & e + f & \quad & c + d & = & g + h \\
a + c & = & i + k & \quad & b + d & = & j + l \\
a + d & = & n + o & \quad & b + c & = & m + p \\
f + g & = & j + k & \quad & e + h & = & i + l \\
e + g & = & m + o & \quad & f + h & = & n + p \\
i + j & = & m + n & \quad & k + l & = & o + p \\
\end{array}
\end{equation*}}
Checking that a section in the support is a member of such a family amounts to assigning $1$ to the variable labelling that section, and $0$ to the other variables in its row.

It suffices to show that these constraints cannot be satisfied over the integers mod 2; this implies that they cannot be satisfied over $\ZZ$, since otherwise such a solution would descend via the homomorphism $\ZZ \rarr \ZZ / 2 \ZZ$.
Of course, this will also show that the cohomology obstruction does not vanish even if we use $\ZZ / 2 \ZZ$ as the coefficient group.

All cases for GHZ have been machine-checked in mod 2 arithmetic, and it has been confirmed that the cohomology obstruction witnesses the impossibility of extending any section in the support to all measurements; thus \textit{cohomology witnesses the strong contextuality of GHZ}.

\section{Kochen-Specker}\label{sec:ks}

We shall now examine covers that can be used for Kochen-Specker arguments. We shall show that the obstructions do not vanish in these cases, providing cohomological proofs of Kochen-Specker  theorems.

We introduce a general notion of Kochen-Specker-type models. We consider two outcomes, $0$ and $1$.
%Let us fix a cover $\UU$ of the measurement set $X$.
For any $C \in \mathcal{U}$, we define $s_{C,m} \in \mathcal{E}(C)$ to be the
section that assigns $1$ to $m$ and $0$ to all other measurements in $C$.
In a Kochen-Specker problem, we wish to assign the outcome $1$ to a single measurement in each context.
Thus, the \myemph{Kochen-Specker support} for the cover $\UU$ is the presheaf given by $S_e(C) = \{s_{C,m} \mid m \in C\}$.

\subsection*{Example: The Triangle}

We shall begin with the simplest Kochen-Specker scenario: the triangle. This has previously been discussed in \cite{abramsky2011unified}, and in a somewhat different context in \cite{Liang20111}. It cannot be realized in quantum mechanics, but it is useful to set the scene.

The triangle is the following cover on three measurements, $A$, $B$, $C$:
\[ \{ A, B \}, \quad \{ B, C \}, \quad \{ A, C \} . \]
We will be interested in the \textit{Kochen-Specker support}, which contains those sections with exactly one $1$ among the outcomes.
Thus we have the following table:
\begin{center}
\renewcommand{\arraystretch}{1.2}
\begin{tabular}{l|ccccc}
& $00$ & $01$ & $10$ & $11$  &  \\ \hline
$AB$ & $0$ & $1$ & $1$ & $0$ & \\
$BC$ & $0$ & $1$ & $1$ & $0$ & \\
$CA$ & $0$ & $1$ & $1$ & $0$ & \\
\end{tabular}
\end{center}
The content of the Kochen-Specker theorem is that there is no compatible family defining a global section within this support. The cohomological statement is that, for any choice of section $s$ in the support, the obstruction $\obst(s)$ does not vanish.

We label the coefficients for a candidate family as follows:
\begin{center}
\renewcommand{\arraystretch}{1.2}
\begin{tabular}{l|ccccc}
& $00$ & $01$ & $10$ & $11$  &  \\ \hline
$AB$ & $0$ & $a$ & $b$ & $0$ & \\
$BC$ & $0$ & $c$ & $d$ & $0$ & \\
$CA$ & $0$ & $e$ & $f$ & $0$ & \\
\end{tabular}
\end{center}
The constraints on the coefficients for a compatible family are as follows:
\[  a = f, \quad b = e, \quad a = d, \quad b = c, \quad d = e, \quad c = f . \]
These equations imply that all the variables are equal.

Checking that a section in the support has a non-vanishing obstruction amounts to setting the variable labelling that section to $1$, and the other variables in its row to $0$.
Clearly there is no solution for any such assignment, which would imply that $1 = 0$.

\subsection*{Example: The 18-Vector Kochen-Specker Configuration}

We look at the 18-vector construction in $\Real^4$ from \cite{cabello1996bell}.
This uses the following measurement cover, where the columns are the sets in the cover.
\begin{center}
\renewcommand{\arraystretch}{1.2}
\begin{tabular}{|c|c|c|c|c|c|c|c|c|} \hline
$A$ & $A$ & $H$ & $H$ & $B$ & $I$ & $P$ & $P$  & $Q$ \\ \hline
$B$ & $E$ & $I$ & $K$ & $E$ & $K$ & $Q$ & $R$ & $R$  \\ \hline
$C$ & $F$ & $C$ & $G$ & $M$ &  $N$ & $D$ & $F$ & $M$  \\ \hline
$D$ & $G$ & $J$ & $L$ & $N$  & $O$ & $J$ & $L$ & $O$  \\ \hline
\end{tabular}
\end{center}

We label the coefficients for a candidate family as follows:
\begin{center}
\renewcommand{\arraystretch}{1.2}
\begin{tabular}{l|cccc}
& $1000$ & $0100$ & $0010$ & $0001$    \\ \hline
$ABCD$ & $a$ & $b$ & $c$ & $d$   \\
$AEFG$ & $a$ & $e$ & $f$ & $g$   \\
$HICJ$ & $h$ & $i$ & $c$ & $j$   \\
$HKGL$ & $h$ & $k$ & $g$ & $l$   \\
$BEMN$ & $b$ & $e$ & $m$ & $n$   \\
$IKNO$ & $i$ & $k$ & $n$ & $o$   \\
$PQDJ$ & $p$ & $q$ & $d$ & $j$ \\
$PRFL$ & $p$ & $r$ & $f$ & $l$ \\
$QRMO$ & $q$ & $r$ & $m$ & $o$
\end{tabular}
\end{center}

Note that some of the constraints on the coefficients take the form of simple equations between variables, allowing us to reduce from 36 to 18 variables; we have used this reduction in the table.

The remaining constraints are expressed by the following equations.
{\renewcommand{\arraystretch}{1.5}\[ 
\begin{array}{lclclcl}
b + c + d & = & e + f + g & \qquad & a + b + d & = & h + i + j \\
a + c + d & = & e + m + n & \qquad & a + b + c & = & p + q + j \\
a + f + g & = & b + m + n & \qquad & a + e + f & = & h + k + l \\
a + e + g & = & p + r + l & \qquad & i + c + j & = & k + g + l \\
h + c + j & = & k + n + o & \qquad & h + i + c & = & p + q + d \\
h + g + l & = & i + n + o & \qquad & h + k + g & = & p + r + f \\
b + e + n & = & q + r + o & \qquad & b + e + m & = & i + k + o \\
i + k + n & = & q + r + m & \qquad & q + d + j & = & r + f + l \\
p + d + j & = & r + m + o & \qquad & p + f + l & = & q + m + o 
\end{array}
\]}

Checking that a section in the support has a non-vanishing obstruction amounts to setting the variable labelling that section to $1$, and the other variables in its row to $0$.

If the equations have no solution for all such assignments, this shows that the cohomology witnesses the contextuality of the model.

This has been machine-checked for mod 2 arithmetic, confirming that we have a \textit{cohomological witness for the Kochen-Specker theorem}.

\subsection*{A Class of Kochen-Specker-type Models}

A necessary condition can be given for Kochen-Specker-type models to have a global section is given in \cite{abramsky2011unified}. 

\begin{proposition}[\cite{abramsky2011unified}]
The existence of a global section implies that 
\[\gcd\{d_m \mid m \in X\} \;\; \mid \;\; |\UU| ,\]
where $d_m := |\{C \in \UU \mid m \in C\}|$.
\end{proposition}

We shall refer to this as the \GCD condition. All models that do not satisfy the \GCD condition are therefore strongly contextual. Using a similar argument, we can show that the cohomology witnesses strong contextuality of any model in this class, as long as we assume a natural connectedness property. In fact, it witnesses strong contextuality of some connected models outside of this class, so it captures the property more finely than the \GCD condition.

A model is said to be \myemph{connected} if, for any contexts $C, C' \in \UU$, one can find a 
a finite sequence of contexts $C_0 = C, C_1, C_2, \ldots, C_n, C_{n+1} = C'$ such that
$\forall \, i \in \{0, \ldots, n\}.\;\; C_i \cap C_{i+1} \neq \emptyset$.

\begin{proposition}
If $\gamma(s)$ vanishes for some section $s$ in the support of a connected Kochen-Specker-type model, then the \GCD condition holds for that model.
\end{proposition}

\begin{proof}
Assume that $\gamma(s_0)=0$  for some section $s_0 \in S_e(C_0)$ in the support.
This means that there is a compatible family $\{r_C \in \FF(C)\}_{C \in \UU}$
of $\ZZ$-linear combinations of sections of $S_e$, with $r_{C_0} = s_0$. 
Recall that the support of each context is $S_e(C) = \{s_{C,m} \mid m \in C\}$. Let $c_{C,m}$ denote the coefficient corresponding to section $s_{C,m}$ in the linear combination $r_C$.

If $m \in C \cap C'$ for contexts
$C = \{m, m_1, \ldots, m_r\}, C' = \{m, m_1', \ldots, m_s'\} \in \UU$,
we get the following cohomology equations:
\[
c_{C,m} = c_{C',m} \qquad\qquad
c_{C,m_1} + \cdots + c_{C,m_r} = c_{C',m_1'} + \cdots + c_{C',m_s'}
\]
By using the equations of the first kind, we can identify all the coefficients
of the form $c_{m,C}$ for the same measurement $m$, in much the same way as we did for the 18-vector Kochen-Specker example. So, we can unambiguously denote these
coefficients by $c_m$ alone.
Summing the two equations above then gives
\[\sum_{x \in C} c_x = c_m + c_{m_1} + \cdots + c_{m_r} = c_m + c_{m_1'} + \cdots + c_{m_s'} = \sum_{x' \in C'} c_{x'}\]
This means that the sums of the coefficients of $r_C$ and $r_{C'}$ are the same.
By connectedness,
and since the sum is equal to $1$ for the context $C_0$ (where we take our starting section $s_0$),
the coefficients must sum to one in every context.

Hence, we have
\[|\UU| = \sum_{C\in\UU} 1 = \sum_{C\in\UU} \; \sum_{m\in C} c_{m} = \sum_{m\in X}d_mc_m =
g \sum_{m\in X} \frac{d_m}{g}c_m\]
where $d_m := |\{C \in \UU \mid m \in C\}|$ as before and $g := \gcd\{d_m \mid m \in X\}$.
Since $g$ divides ${d_m}$ for all $m$, we conclude that $g$ divides $|\UU|$.
\end{proof}

\section{Example: The Peres-Mermin Square}
We now turn to an important example, the Peres-Mermin square \cite{peres1990incompatible,mermin1990simple}, which can be realized in quantum mechanics using two-qubit observables.

The structure of the square is as follows:
\begin{center}
\renewcommand{\arraystretch}{1.5}
\begin{tabular}{|c|c|c|}
\hline
$A$ & $B$ & $C$ \\ \hline
$D$ & $E$ & $F$ \\ \hline
$G$ & $H$ & $I$ \\ \hline
\end{tabular}
\end{center}
The compatible families of measurements are the rows and columns of this table.
The problem in question differs from the usual Kochen-Specker type problems
in that we don't ask for exactly one $1$ at each maximal context. Instead, we ask
that each `row context' has an odd number of $1$s whereas each `column context'
has an even number $1$s.

Hence, the support table is the following.
Note that the first three lines correspond to the row contexts and the remaining
three to the column contexts.
\begin{center}
\renewcommand{\arraystretch}{1.5}
\begin{tabular}{l|cccccccc}
      & $000$ & $001$ & $010$ & $011$ & $100$ & $101$ & $110$ & $111$   \\ \hline
$ABC$ &  $0$  &  $1$  &  $1$  &  $0$  &  $1$  &  $0$  &  $0$  &  $1$  \\
$DEF$ &  $0$  &  $1$  &  $1$  &  $0$  &  $1$  &  $0$  &  $0$  &  $1$  \\
$GHI$ &  $0$  &  $1$  &  $1$  &  $0$  &  $1$  &  $0$  &  $0$  &  $1$  \\
$ADG$ &  $1$  &  $0$  &  $0$  &  $1$  &  $0$  &  $1$  &  $1$  &  $0$  \\
$BEH$ &  $1$  &  $0$  &  $0$  &  $1$  &  $0$  &  $1$  &  $1$  &  $0$  \\
$CFI$ &  $1$  &  $0$  &  $0$  &  $1$  &  $0$  &  $1$  &  $1$  &  $0$  \\
\end{tabular}
\end{center}

We display the coefficients for a candidate family as follows.
\def\na{\overline{a}}
\def\nb{\overline{b}}
\def\nc{\overline{c}}
\def\nt{\overline{t}}
\begin{center}
\renewcommand{\arraystretch}{1.5}
\begin{tabular}{l|cccccccc}
      & $000$ & $001$ & $010$ & $011$ & $100$ & $101$ & $110$ & $111$   \\ \hline
$ABC$ &  $0$  & $c_1$ & $b_1$ &  $0$  & $a_1$ &  $0$  &  $0$  & $t_1$  \\
$DEF$ &  $0$  & $c_2$ & $b_2$ &  $0$  & $a_2$ &  $0$  &  $0$  & $t_2$  \\
$GHI$ &  $0$  & $c_3$ & $b_3$ &  $0$  & $a_3$ &  $0$  &  $0$  & $t_3$  \\
$ADG$ &$\nt_4$&  $0$  &  $0$  &$\na_4$&  $0$  &$\nb_4$&$\nc_4$&  $0$  \\
$BEH$ &$\nt_5$&  $0$  &  $0$  &$\na_5$&  $0$  &$\nb_5$&$\nc_5$&  $0$  \\
$CFI$ &$\nt_6$&  $0$  &  $0$  &$\na_6$&  $0$  &$\nb_6$&$\nc_6$&  $0$  \\
\end{tabular} 
\end{center}
The equations expressing the constraints are the following:
{\renewcommand{\arraystretch}{1.5}
\begin{equation*}
\label{new_pmeqs}
\begin{array}{lclclcl}
  a_1 + t_1 &=& \nb_4 + \nc_4 & \quad & \na_4 + \nt_4  &=& b_1 + c_1 
  \\
  b_1 + t_1 &=& \nb_5 + \nc_5 & \quad & \na_5 + \nt_5  &=& a_1 + c_1 
  \\
  c_1 + t_1 &=& \nb_6 + \nc_6 & \quad & \na_6 + \nt_6  &=& a_1 + b_1 
  \\
  ~\\
  a_2 + t_2 &=& \na_4 + \nc_4 & \quad & \nb_4 + \nt_4  &=& b_2 + c_2 
  \\
  b_2 + t_2 &=& \na_5 + \nc_5 & \quad & \nb_5 + \nt_5  &=& a_2 + c_2 
  \\
  c_2 + t_2 &=& \na_6 + \nc_6 & \quad & \nb_6 + \nt_6  &=& a_2 + b_2 
  \\
  ~\\
  a_3 + t_3 &=& \na_4 + \nb_4 & \quad & \nc_4 + \nt_4  &=& b_3 + c_3 
  \\
  b_3 + t_3 &=& \na_5 + \nb_5 & \quad & \nc_5 + \nt_5  &=& a_3 + c_3 
  \\
  c_3 + t_3 &=& \na_6 + \nb_6 & \quad & \nc_6 + \nt_6  &=& a_3 + b_3 
  \\
\end{array}
\end{equation*}}

We start by choosing a section $s$. We set its coefficient to $1$, and the coefficients of all other sections in the same context to $0$. Then a solution to the equations above would give a compatible family in $\FF$ containing $s$, meaning that the cohomological obstruction vanishes.
It has been machine-checked using mod-2 arithmetic that 
there is no solution to the system for any choice of starting section $s$.

\section{Limitations and Further Directions}

There are two immediate limitations to the results we have described:
\begin{itemize}
\item The cohomological condition for contextuality is sufficient, but not necessary.
It is interesting to note that the example where a false positive does arise here, namely the Hardy model, is non-local and hence contextual, but not \textit{strongly contextual}.

It has been possible to construct a strongly contextual model for which a false positive does arise. This is the Kochen-Specker model for the cover
\[ \{A,B,C\}, \{B,D,E\}, \{C,D,E\}, \{A,D,F\}, \{A,E,G\} \]
However, unlike our earlier examples, this model does not satisfy any reasonable criterion for symmetry, nor does it satisfy any strong form of connectedness.
In fact, the existence of measurements belonging to a single maximal context, namely $F$ and $G$,
seems to be crucial in this example. It is always possible to choose coefficients
for $s_{\{A,D,F\},F}$ and $s_{\{A,E,G\},G}$ (using the notation of section \ref{sec:ks}) that will make the coefficients of the respective contexts sum to one, without imposing constrains on the other contexts.

\begin{conjecture}
Under suitable assumptions of symmetry and connectedness, the cohomology obstruction is a complete invariant for strong contextuality.
\end{conjecture}

In \cite{vorobev}, Vorob'ev
characterised the covers (or to be more precise the simplicial complexes these generate)
for which any model is extendable; i.e. non-contextual.
These are exactly the complexes which can be reduced to an empty complex by removing certain extremal
maximal contexts. From the proof of this result, one can see that the non-extendability of a model would be already noticed
in its reduced version, which allows us to focus on witnessing non-contextuality for irreducible covers.
A necessary condition for a context to be extremal
is that it possesses measurements not belonging to any other maximal context. Even though the above example has no
extremal contexts, and thus is irreducible, it does have this weaker property.

\item Thus far, we have simply been computing the obstructions by brute force enumeration, so the results we have obtained can only be considered a proof of concept. What one would like is to use the machinery of homological algebra and exact sequences to obtain more conceptual  and general results. 
\end{itemize}

Overcoming these limitations is the main objective for future work. This may require refining the abelian presheaf $\FF$ to yield a finer invariant.

\section*{Acknowledgements}
Support from EPSRC Senior Research Fellowship EP/E052819/1, the U.S. Office of Naval Research Grant Number N000141010357, 
 the Marie Curie Initial Training Network - MALOA - From MAthematical LOgic to Applications, PITN-GA-2009-238381, and the National University of Ireland Travelling Studentship Program is gratefully acknowledged.

\bibliographystyle{eptcs}
\bibliography{bdbib-1}
\end{document}